\documentclass[%
 reprint,
 amsmath,amssymb,
 aps,
 prx,
]{revtex4-2}

\usepackage{graphicx}
\usepackage{dcolumn}
\usepackage{amsmath, amssymb, amsthm, bm, amsfonts, mathrsfs, bbm}
\usepackage{dsfont}
\usepackage{physics}
\usepackage{xcolor}
\usepackage[colorlinks=true,linkcolor=blue,citecolor=blue,urlcolor=blue]{hyperref}

\newtheorem{theorem}{Theorem}

\newtheorem{lemma}{Lemma}

\newtheorem{proposition}{Proposition}

\begin{document}

\title{Quantum tomography for non-iid sources}

\author{Leonardo Zambrano}
\email{leonardo.zambrano@icfo.eu}
\affiliation{ICFO - Institut de Ci\`encies Fot\`oniques, The Barcelona Institute of Science and Technology, 08860 Castelldefels, Barcelona, Spain}

\begin{abstract}
Quantum state and process tomography are typically analyzed under the assumption that devices emit independent and identically distributed (i.i.d.) states or channels. In realistic experiments, however, noise, drift, feedback, or adaptive behavior can violate this assumption. We show that projected least-squares tomography remains statistically optimal even under fully adaptive state and channel preparation. Specifically, we prove that the sample complexity for reconstructing the time-averaged state or channel matches the optimal i.i.d. scaling up to logarithmic factors for non-adaptive, single-copy measurements. For average states of rank $r$, the sample complexity is $\tilde{\mathcal{O}}(d r^2/\epsilon^2)$ to achieve trace-norm accuracy $\epsilon$, while for process tomography it is $\tilde{\mathcal{O}}(d^6/\epsilon^2)$ to achieve diamond-norm accuracy $\epsilon$. Thus, dropping the i.i.d.\ assumption does not increase the fundamental sample complexity of quantum tomography, but only changes the interpretation of the reconstructed object. Finally, by tracking the drift of the source using only past outcomes, we certify high average preparation purity and entanglement even when the time-averaged state is mixed or separable.
\end{abstract}

\maketitle

\section{Introduction}

Quantum state tomography (QST) is one of the central primitives of quantum information science. It provides a systematic procedure to reconstruct an unknown quantum state from measurement data and underlies experimental validation in quantum computing, quantum communication, and quantum sensing~\cite{hradil1997quantum, gross2010quantum, o2015efficient, haah2017sample, anshu2024survey}. In its standard formulation, QST assumes that the source emits $N$ independent and identically distributed (i.i.d.) copies of a fixed quantum state $\rho$. Under this assumption, measurement outcomes are independent random variables drawn from a fixed distribution, and powerful concentration tools, such as matrix Bernstein inequalities~\cite{tropp2012user}, can be used to derive finite-sample guarantees. In particular, for single-copy, non-adaptive complex projective $2$-design measurements, projected least-squares (PLS) quantum tomography \cite{guta2020fast, surawy2022projected, zambrano2025fast} achieves near-optimal sample complexity scaling as $N = \tilde{\mathcal{O}}\!\left({d r^{2}}/{\epsilon^{2}}\right)$ for estimating a rank-$r$ state with trace-norm accuracy $\epsilon$ \cite{haah2017sample}.

However, the i.i.d.\ assumption is rarely satisfied in realistic experiments. Quantum devices drift in time due to thermal fluctuations, laser instabilities, calibration errors, and slow environmental noise~\cite{klimov2018fluctuations, proctor2020detecting}. In adaptive experiments, the state prepared at time $t$ may depend explicitly on previous measurement outcomes through feedback or control loops~\cite{white2020demonstration, mcewen2021removing}. In more adversarial scenarios, the preparation procedure may react strategically to previously revealed information~\cite{pirandola2020advances, fawzi2024learning}. In all these cases, the sequence of emitted states $\rho_{1}, \ldots, \rho_{N}$ cannot be modeled as identical copies of a single density matrix. Instead, the source generates a random trajectory of states that may depend arbitrarily on the past.

For such a trajectory, the result of tomography depends not only on the prepared states, but also on how the measurements sample them~\cite{van2013quantum}. Consider a qubit whose Bloch vector rotates uniformly from $+z$ to $-z$ in the $xz$ plane during the experiment, so that its $Z$ expectation vanishes on average (Fig.~\ref{fig:measurement_order}). Suppose that $X$, $Y$, and $Z$ are measured in three consecutive blocks of equal length to perform tomography. The $Z$ block then only sees the last third of the rotation, and the reconstruction predicts $\langle Z\rangle\simeq-0.83$. Measuring in the reverse order gives $\langle Z\rangle\simeq+0.83$. The same preparations thus lead to incompatible reconstructions, depending only on the order of the measurements. What, then, does tomography reconstruct when the source is not i.i.d.?

The answer is simple if the measurement setting is drawn at random in every round. PLS tomography then reconstructs the average of the states prepared during the run, $\bar\rho_N=\frac{1}{N}\sum_{t=1}^N\rho_t$, even when each preparation depends on the earlier history of the experiment. This average characterizes the observed run, while certifying unmeasured systems requires further assumptions about the source or a different sampling protocol \cite{navarro2026certifying}. Its operational meaning follows from a simple identity. For every observable $O$, such as a Hamiltonian or an entanglement witness, $\tr(O\bar\rho_N)=\frac{1}{N}\sum_{t=1}^N\tr(O\rho_t)$. If the source does not react to the settings and outcomes, this is the mean value that measuring $O$ on every emitted system would have given.

Our main result shows that PLS quantum state tomography remains statistically optimal even under fully adaptive state preparation. Informally, for measurements given by a complex projective $2$-design in dimension $d$, the number of samples required to estimate the time-averaged state with trace-norm error at most $\epsilon$ scales optimally (up to logarithmic factors) as $N = \tilde{\mathcal{O}}\!\left({d r^{2}}/{\epsilon^{2}} \right)$, where $r$ is the rank of $\bar{\rho}_N$. For Pauli basis measurements on $n$ qubits, the sample complexity scales as $N = \tilde{\mathcal{O}}\!\left({3^{n} r^{2}}/{\epsilon^{2}}\right)$, matching the known i.i.d.\ PLS rates~\cite{guta2020fast}. Furthermore, by leveraging the Choi--Jamio{\l}kowski isomorphism, we demonstrate that this framework extends naturally to quantum process tomography, allowing us to estimate a time-averaged channel $\bar{\mathcal{E}}_N$ with an optimal sample complexity of $N = \tilde{\mathcal{O}}\!\left(d^6/\epsilon^2\right)$ in diamond norm~\cite{surawy2022projected, oufkir2023sample}.  Thus, remarkably, dropping the i.i.d.\ assumption does not increase the fundamental scaling of PLS tomography.

Beyond linear quantities, the data contain information that averaging discards. For any convex function $g$, such as purity or negativity~\cite{watrous2018theory}, convexity gives $\frac{1}{N}\sum_t g(\rho_t)\ge g(\bar\rho_N)$. This lower bound can be loose: a qubit with a drifting phase
can have a mixed average even when every preparation is pure. We show that following the source recovers part of this information. By rotating the estimates into a frame learned from earlier outcomes, we certify purity and entanglement hidden by averaging, even when the phase drifts as a random walk. These certificates require no model of the drift, since the quality of tracking only affects how tight they are. They cannot, however, be turned into estimates in general: we prove that for any non-affine $g$, no protocol receiving one copy of each state, even one using collective measurements, can estimate $\frac{1}{N}\sum_t g(\rho_t)$ with vanishing worst-case error.

The remainder of this paper is structured as follows. In Section~\ref{sec:state}, we present the PLS protocol for quantum state tomography, discuss the meaning of the reconstructed average, and establish our main non-asymptotic guarantees. The same section introduces frame tracking and the limits on estimating nonlinear quantities. In Section~\ref{sec:channel}, we extend the framework to quantum process tomography. Section~\ref{sec:numerics} illustrates the results numerically, and Section~\ref{sec:conclusions} concludes. Proofs are given in the appendices.

\begin{figure}[t]
    \centering
    \includegraphics[width=0.9\columnwidth]{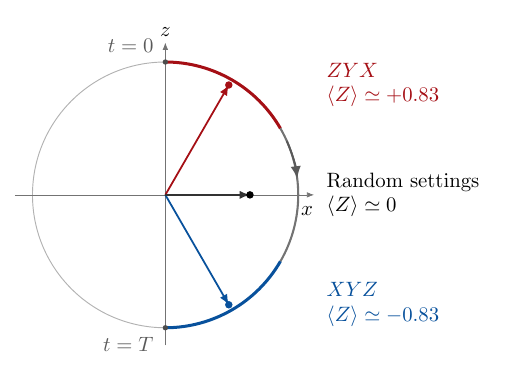}
    \caption{Tomography of a qubit rotating from $+z$ to $-z$ in the $xz$ plane. The $ZYX$ and $XYZ$ block schedules yield incompatible reconstructions,
    whereas randomized settings recover the time-averaged  state. Colored arcs indicate when each schedule measures $Z$. Points show the estimated states.}
    \label{fig:measurement_order}
\end{figure}

\section{Robust projected least-squares quantum state tomography}\label{sec:state}

\subsection{Non-i.i.d.\ physical model}

We consider a scenario in which a quantum source generates a sequence of $N$ quantum states, but we drop the standard assumption that these states are independent and identically distributed. Instead, we model the experiment as an interaction between an experimenter and a potentially adaptive source.

The experiment proceeds in rounds indexed by $t = 1, \dots, N$. At each round $t$, the source prepares a quantum state $\rho_t \in \mathcal{D}(\mathbb{C}^d)$, which is then measured. We allow the source to be fully adaptive: the state $\rho_t$ may depend arbitrarily on everything that is fixed before the measurement of round $t$, which we denote $\mathcal{F}_{t-1}$. This history includes all previous measurement settings and outcomes, as well as hidden parameters affecting the preparation, such as thermal fluctuations or laser-phase drift. Thus $\rho_t$ is determined by $\mathcal{F}_{t-1}$, but it need not be known to the experimenter, since part of the history is inaccessible. The model also covers sources with quantum memory, whose emitted systems may be entangled across rounds. In that case, $\rho_t$ is the reduced state of the $t$-th system conditioned on the classical data of the earlier rounds.

Consequently, there is no single underlying target state $\rho$. Instead, the experiment generates a trajectory of states $\rho_1, \dots, \rho_N$, and our objective is to estimate
\begin{align}
    \bar{\rho}_N = \frac{1}{N} \sum_{t=1}^{N} \rho_t.
\end{align}

Once $\rho_t$ is prepared, an informationally complete (IC) measurement described by a POVM $\{\Pi_k\}_{k=1}^{M}$ is performed, defining a random variable $Y_t$. If the measurement is implemented by choosing a setting at
random, $Y_t$ includes both the setting and the measurement result.
The setting must be drawn independently of $\mathcal{F}_{t-1}$, so the source cannot adapt its preparation to it. The probability of each outcome is then given by
\begin{align}
P(Y_t = k \mid \mathcal{F}_{t-1})
= P(Y_t = k \mid \rho_t)
= \tr(\Pi_k \rho_t).
\label{eq:born}
\end{align}
The source may choose $\rho_t$ adaptively, but it cannot influence the outcome beyond the probabilities imposed by quantum mechanics. This is the only structure our statistical analysis requires.

As in standard tomography, we assume the POVM is known and calibrated \cite{zambrano2025fast}. Tomographic data alone cannot in general distinguish a change in the source from one in the detector. Distinguishing the two requires additional calibration information or a trusted reference~\cite{nielsen2021gate}.

\subsection{What the average state represents}
\label{sec:meaning}

For every observable $O$ and trajectory $\rho_1, \dots, \rho_N$,
\begin{align}
    \tr(O\bar\rho_N)
    =\frac{1}{N}\sum_{t=1}^N\tr(O\rho_t),
\end{align}
so $\bar\rho_N$ gives the run average of any single-system quantity that is linear in the state, such as the fidelity with a fixed pure target state. By Theorem~\ref{thm:first_moment_converse}, these are the only quantities whose run averages can be estimated from one copy of each state for arbitrary trajectories. Nonlinear quantities, in contrast, are not determined by the data. For example, after averaging over the unknown preparation labels, a source emitting independent Haar-random pure states produces the same joint state, $(\mathds{1}/d)^{\otimes N}$, as one emitting $\mathds{1}/d$ in every round, so no measurement can reveal that the preparations are pure.

For sources with quantum memory, $\bar\rho_N$ averages states conditioned on the preceding outcomes, and therefore depends on the measurements performed. If the qubits of an $N$-qubit GHZ state are emitted one at a time and measured in random Pauli bases, $X$ and $Y$ outcomes leave the remaining single-qubit marginals at $\mathds{1}/2$, but the first $Z$ outcome $b \in \{0, 1\}$ collapses all remaining qubits to $\ket{b}$. Hence, for large $N$, $\bar\rho_N\approx\ketbra{b}$, although every single-qubit marginal is $\mathds{1}/2$.

\subsection{Estimation protocol}

Our protocol operates in the single-copy regime with non-adaptive measurements. To perform the tomography, we consider two classes of measurement ensembles. First, we employ POVMs $\{\frac{d}{M} P_k \}_{k=1}^{M}$ such that the projectors $\{P_k\}_{k=1}^{M}$ form a complex projective $2$-design in dimension $d$. Examples include symmetric informationally complete POVMs (SIC-POVMs) \cite{renes2004symmetric} and complete sets of mutually unbiased bases (MUBs) \cite{klappenecker2005mutually}. Second, for $n$-qubit systems ($d=2^n$), we consider tensor products of single-qubit $2$-designs \cite{dankert2009exact}. Each POVM element is of the form $\frac{2^n}{m^n} P_{k_1} \otimes \dots \otimes P_{k_n}$, where each local set $\{P_{k_j}\}_{k_j=1}^{m}$ forms a single-qubit $2$-design. A canonical example is the Pauli measurement scheme, in which the local ensemble consists of the $m=6$ eigenstates of the operators $\{X, Y, Z\}$.

The estimation procedure consists of three steps \cite{guta2020fast}.

\begin{enumerate}
    \item \textit{Data Collection.}
    For each round $t=1, \dots, N$, the source prepares a state $\rho_t$, which is then measured. This yields a dataset $\{Y_1, \dots, Y_N\}$, where each $Y_t$ is distributed according to
    \[
    P(Y_t = k \mid \mathcal{F}_{t-1}) = \tr(\Pi_k \rho_t).
    \]

    \item \textit{Linear Estimation.}
    Each outcome $Y_t$ is mapped to a single-shot matrix $\hat{\rho}_t$.

    For global $2$-designs, if $Y_t = k$, then
    \begin{equation}
        \hat{\rho}_t = (d+1) P_k - \mathds{1}.
    \end{equation}

    For local $2$-designs, if $Y_t = k_1 \dots k_n$,
    \begin{equation}
        \hat{\rho}_t = \bigotimes_{j=1}^n \left( 3 P_{k_j} - \mathds{1} \right).
    \end{equation}

    We compute the time-averaged linear estimator
    \begin{equation}
        \hat{L}_N = \frac{1}{N} \sum_{t=1}^{N} \hat{\rho}_t.
    \end{equation}

    \item \textit{Projection.}
    $\hat{L}_N$ has unit trace but is not guaranteed to be positive semidefinite. The final physical estimator is obtained by projecting $\hat{L}_N$ onto the set of density matrices \cite{smolin2012efficient}:
    \begin{align}
        \hat{\rho}_{\mathrm{PLS}}
        =
        \operatorname*{argmin}_{\omega \in \mathcal{D}(\mathbb{C}^d)}
        \| \hat{L}_N - \omega \|_F.
    \end{align}
\end{enumerate}

Throughout, we call $\hat\rho_t$ the single-shot estimator, $\hat L_N$ the linear estimator, and $\hat\rho_{\mathrm{PLS}}$ the PLS estimator. For the measurements considered here, the least-squares estimator built from the empirical outcome frequencies coincides with the linear estimator $\hat L_N$ (Appendix~\ref{app:estimators}).

The theoretical robustness of this protocol relies on the construction of the linear estimator. The defining feature of the single-shot estimator $\hat{\rho}_t$ is that it is conditionally unbiased. Since the measurement outcome is independent of the prior history $\mathcal{F}_{t-1}$ given $\rho_t$, we have
\begin{equation}
    \mathbb{E}\!\left[ \hat{\rho}_t \mid \mathcal{F}_{t-1} \right] = \mathbb{E}_{k \sim p(\rho_t)} \left[ \hat{\rho}_t \right] = \rho_t.
\end{equation}
This equality holds regardless of how $\rho_t$ was prepared. If we define the estimation error at step $t$ as
\begin{align}
    X_t = \hat{\rho}_t - \rho_t,
\end{align}
then $\mathbb{E}[X_t \mid \mathcal{F}_{t-1}] = 0$. Physically, this means that although the source may prepare $\rho_t$ adaptively based on the entire past history, it cannot systematically bias the expectation of $X_t$, since it cannot force a specific measurement outcome to occur. Once $\rho_t$ is fixed, the outcome is governed solely by Born’s rule, ensuring that the estimator $\hat{\rho}_t$ remains centered, in conditional expectation, on the true state. The total estimation error accumulated over $N$ rounds is therefore a sum of conditionally zero-mean fluctuations. This structure allows us to go beyond standard i.i.d.\ concentration inequalities (such as Matrix Bernstein) and instead apply the Matrix Freedman inequality, which controls cumulative deviations in terms of their predictable variance (see Appendix~\ref{app:tools}).

Finally, the estimator $\hat{L}_N$ can also be computed from aggregate statistics, as in Ref.~\cite{guta2020fast}. Let $\{\hat{f}_k\}_{k=1}^M$ denote the empirical relative frequencies of the different measurement outcomes. By linearity, $\hat{L}_N = \sum_{k=1}^M \hat{f}_k \hat{\rho}_k$. Thus, in an experiment, it suffices to store the histogram of outcome counts rather than the full sequence of outcomes. 

\subsection{Non-asymptotic guarantees}

A key feature of PLS tomography is its automatic adaptation to the effective rank of the state. For any state $\rho$, we define the residual spectral mass as the sum of the smallest $d-r$ eigenvalues: $\Sigma_r(\rho) = \sum_{j=r+1}^d \lambda_j(\rho)$. If $\rho$ has rank at most $r$, then $\Sigma_r(\rho)=0$.

Our analysis has two independent parts: a concentration bound for $\hat{L}_N$ around $\bar{\rho}_N$, and a deterministic property of the projection \cite{guta2020fast}. Proofs are given in Appendix~\ref{app:proofs}.

\begin{proposition}
\label{prop:concentration}
For any adaptive source and $\tau>0$,
\begin{align}
    P\!\left(\|\hat{L}_N-\bar\rho_N\|\ge\tau\right)
    \le 2d\exp\!\left(-\frac{N\tau^2}{2(v+R\tau/3)}\right),
    \label{eq:concentration}
\end{align}
where $\|\cdot\|$ is the operator norm, $(v,R)=(2d,d)$ for global $2$-design measurements, and $(v,R)=(3^n,2^n)$ for tensor products of local $2$-designs.
\end{proposition}

\begin{lemma}
\label{lem:norm_conversion}
Let $L$ be a unit-trace Hermitian operator, and let $\hat\rho$ be its projection onto $\mathcal{D}(\mathbb{C}^d)$ in Frobenius norm. If $\|L-\sigma\|\le\tau$ for some state $\sigma$, then $\|\hat\rho-\sigma\|\le2\tau$ and, for every $r$,
\begin{align}
    \|\hat\rho-\sigma\|_1
    \le 4r\tau+2\min\left\{\Sigma_r(\sigma),\Sigma_r(\hat\rho)\right\}.
\end{align}
\end{lemma}

Lemma~\ref{lem:norm_conversion} follows from Proposition~1 and its proof in Ref.~\cite{guta2020fast}. This lemma is deterministic and therefore applies without independence assumptions. Combining it with Proposition~\ref{prop:concentration} at $\tau=\epsilon/(4r)$ gives the following guarantee for the PLS estimator.

\begin{theorem}\label{thm:main_result}
Let the source emit an arbitrary adaptive sequence of states with time-average $\bar{\rho}_N$. Fix a target rank $r$, accuracy $\epsilon > 0$, and failure probability $\delta \in (0,1)$.

With probability at least $1-\delta$, the estimation error satisfies
\begin{equation}
    \|\hat{\rho}_{\mathrm{PLS}} - \bar{\rho}_N\|_1
    \le
    \epsilon
    +
    2 \min \left\{
        \Sigma_r(\bar{\rho}_N),
        \Sigma_r(\hat{\rho}_{\mathrm{PLS}})
    \right\},
\end{equation}
provided that the number of samples $N$ satisfies:\\
For global complex projective $2$-design measurements,
    \begin{equation}
        N \ge
        \frac{64\, d\, r^2 (1 + \epsilon/24r)}{\epsilon^2}
        \log\!\left(\frac{2d}{\delta}\right).
    \end{equation}
For tensor products of local $2$-designs (random Pauli basis measurements with $d=2^n$),
    \begin{equation}
        N \ge
        \frac{32\, r^2 \bigl(3^n + 2^n \epsilon/12r\bigr)}{\epsilon^2}
        \log\!\left(\frac{2^{n+1}}{\delta}\right).
    \end{equation}

\end{theorem}

If $\bar\rho_N$ has rank at most $r$, the residual spectral mass vanishes and the trace-norm error is at most $\epsilon$. For approximately low-rank states, the additional error is controlled by the weight outside the largest $r$ eigenvalues. The tail $\Sigma_r(\hat\rho_{\mathrm{PLS}})$ can be evaluated directly from the reconstructed state.

For an average state of rank at most $r$, Theorem~\ref{thm:main_result} gives sample scalings of $\tilde{\mathcal{O}}(dr^2/\epsilon^2)$ for global $2$-designs and $\tilde{\mathcal{O}}(3^n r^2/\epsilon^2)$ for local $2$-designs. These match the corresponding scalings of i.i.d.\ PLS tomography with the same measurements~\cite{haah2017sample,guta2020fast}. Even the constants are comparable: Ref.~\cite{guta2020fast} gives sufficient sample counts $N\ge86\,dr^2\epsilon^{-2}\log(d/\delta)$ for $2$-designs and $N\ge43\cdot3^nr^2\epsilon^{-2}\log(d/\delta)$ for Pauli bases.  The guarantee controls the distance to the average generated during the observed run, which may itself be random.

\subsection{Following a drifting source}
\label{sec:tracking}

The average state can give one-sided information about some nonlinear properties of the preparations. For any convex function $g$,
\begin{align}
    \frac{1}{N}\sum_{t=1}^N g(\rho_t)
    \ge g(\bar\rho_N).
    \label{eq:plain_jensen}
\end{align}
Thus, the purity of $\bar\rho_N$ lower bounds the average purity, and any convex entanglement measure evaluated on $\bar\rho_N$ lower bounds its average over the preparations. These bounds can become loose when the source changes too aggressively. For example, a qubit whose orientation changes, as in Fig.~\ref{fig:measurement_order}, can have a mixed average even if every emitted state is pure. To strengthen the bounds, we can use earlier outcomes to estimate the changing orientation and align the single-shot estimators before averaging.

Specifically, before round $t$, choose a unitary $V_t$ from the observed history, for example a rotation that compensates for the drift estimated so far. Define
\begin{align}
    \tilde L_N
    &=\frac{1}{N}\sum_{t=1}^N
      V_t\hat\rho_tV_t^\dagger,
    &
    \tilde\rho_N
    &=\frac{1}{N}\sum_{t=1}^N
      V_t\rho_tV_t^\dagger,
    \label{eq:tracked}
\end{align}
and let $\hat\rho^V$ be the Frobenius-norm projection of $\tilde L_N$ onto $\mathcal{D}(\mathbb{C}^d)$. The rotations act only on the data and leave the measurement protocol unchanged.

\begin{proposition}
\label{prop:tracking}
Let each $V_t$ be a unitary chosen from the observed
history before round $t$. Then
Proposition~\ref{prop:concentration} and
Theorem~\ref{thm:main_result} hold with
$(\hat L_N,\bar\rho_N,\hat\rho_{\mathrm{PLS}})$ replaced by
$(\tilde L_N,\tilde\rho_N,\hat\rho^V)$, with the same
constants. Moreover, if $g$ is convex and satisfies
$g(V_t\rho V_t^\dagger)=g(\rho)$ for every $t$ and $\rho$, then
\begin{align}
    \frac{1}{N}\sum_{t=1}^N g(\rho_t)
    \ge g(\tilde\rho_N).
    \label{eq:tracked_jensen}
\end{align}
\end{proposition}

The proof is in Appendix~\ref{app:newproofs}. Tracking can therefore improve the certificate without requiring a model of the drift. It is particularly useful when the source changes appreciably during the full experiment but remains sufficiently stable over the shorter interval needed to estimate its orientation. The certificate remains valid even when tracking is inaccurate, although it may become less informative.

For a fixed sample size $N$, choose $\tau$ using Proposition~\ref{prop:concentration} so that $\|\tilde L_N-\tilde\rho_N\|\le\tau$ with probability at least $1-\delta$. All certificates below hold on this event.

For the purity, we can obtain a certificate directly from the linear estimator. On the event
$\|\tilde L_N-\tilde\rho_N\|\le\tau$, the inequality
$\tr [(\tilde\rho_N-\tilde L_N)^2]\ge0$ gives
\begin{align}
    \tr(\tilde\rho_N^2)
    &\ge 2\tr(\tilde L_N\tilde\rho_N)-\tr(\tilde L_N^2)
    \nonumber\\
    &\ge \tr(\tilde L_N^2)-2\tau\|\tilde L_N\|_1.
    \label{eq:purity_average}
\end{align}
Together with Eq.~\eqref{eq:tracked_jensen}, this yields a bound which can certify high average purity even when the unrotated average is strongly mixed:
\begin{align}
    \frac{1}{N}\sum_{t=1}^N\tr(\rho_t^2)
    \ge \tr(\tilde L_N^2)-2\tau\|\tilde L_N\|_1.
    \label{eq:purity_certificate}
\end{align}

The same idea applies to entanglement when the rotations are local, $V_t=V_t^A\otimes V_t^B$. Since these rotations preserve entanglement, any convex entanglement measure of $\tilde\rho_N$ lower bounds its average over the preparations. A statistical certificate follows by minimizing the measure over the confidence region \cite{zambrano2024certification, zambrano2026certification}
\begin{align}
    \mathcal{C}
    =
    \left\{
        \omega\succeq0:
        \tr\omega=1,\;
        \|\omega-\tilde L_N\|\le\tau
    \right\}.
\end{align}
For the negativity,
$\mathcal{N}(\omega)=(\|\omega^{T_B}\|_1-1)/2$,
where $T_B$ denotes the partial transpose on subsystem $B$, this amounts to the semidefinite program \cite{vidal2002computable}
\begin{align}
    \frac{1}{N}\sum_{t=1}^N\mathcal{N}(\rho_t)
    \ge
    \min_{\omega,P,M}\quad &\tr M
    \nonumber\\
    \text{s.t.}\quad
    &\omega\in\mathcal{C}, \qquad P,M\succeq0
    \nonumber\\
    &\omega^{T_B}=P-M.
    \label{eq:negativity_sdp}
\end{align}
The same confidence region can be used to certify different entanglement measures chosen after the data have been collected.

\subsection{Estimating nonlinear properties}
\label{sec:trajectory_scope}

The average state gives lower bounds on the average purity and entanglement of the preparations. Obtaining a uniformly accurate estimate of these averages is a stronger task. The example in Sec.~\ref{sec:meaning} already shows the difficulty for purity: different preparation procedures can produce indistinguishable measurement data while having different average purities. We now establish the general limitation illustrated by this example.

For a function $g:\mathcal{D}(\mathbb{C}^d)\to\mathbb{R}$, define
\begin{align}
    G_N(\rho_1,\ldots,\rho_N)
    =\frac{1}{N}\sum_{t=1}^N g(\rho_t).
    \label{eq:G_N}
\end{align}

\begin{theorem}
\label{thm:first_moment_converse}
Let $G_N$ be defined as in Eq.~\eqref{eq:G_N}. If $g$ is affine, then $G_N=g(\bar\rho_N)$ and can be estimated to arbitrary accuracy with arbitrarily high confidence, uniformly over all trajectories, from one copy of each state.

If $g$ is non-affine, there exist constants $\eta_g,c_g>0$ such that every estimator
$\widehat G_N$ obtained from one copy of each state, including by arbitrary collective measurements on $\rho_1\otimes\cdots\otimes\rho_N$, satisfies
\begin{align}
    \sup_{\rho_1,\ldots,\rho_N}
    \Pr\!\left(
        |\widehat G_N-G_N|>\eta_g
    \right)
    \ge \frac{1}{2}-e^{-c_gN}.
\end{align}
Thus, arbitrarily accurate estimation with arbitrarily high confidence over unrestricted trajectories is possible if and only if $g$ is affine.
\end{theorem}

The proof is given in Appendix~\ref{app:nonlinear_proof}. The limitation concerns the average of $g$ over the individual preparations. For example, the purity of the average state, $\tr(\bar\rho_N^2)$, can still be estimated through tomography. Tracking can provide informative lower bounds on the preparation averages when the drift can be followed, even though accurate estimation over unrestricted trajectories is impossible.

The distinction becomes concrete if additional copies are available. Suppose each round supplies two systems in the state $\rho_t\otimes\rho_t$. The swap operator then has expectation $\tr(\rho_t^2)$, making purity a linear observable of the pair~\cite{ekert2002direct}. Its run average can therefore be estimated even when the source changes arbitrarily between rounds. More generally, $k$ copies of each preparation give access to polynomial functionals of degree at most $k$ through observables on $\rho_t^{\otimes k}$.

\section{Robust projected least-squares quantum process tomography}\label{sec:channel}

We now extend the robust PLS framework to the tomography of time-varying quantum channels. Using the Choi–Jamiołkowski isomorphism, we reduce this problem to estimating a bipartite quantum state.

Consider a quantum channel $\mathcal{E}: \mathcal{L}(\mathbb{C}^{d}) \to \mathcal{L}(\mathbb{C}^{d})$ acting on a system of dimension $d$. This channel is completely characterized by its normalized Choi matrix $\rho_{\mathcal{E}} \in \mathcal{D}(\mathbb{C}^{d} \otimes \mathbb{C}^{d})$, defined as:
\begin{equation}
    \rho_{\mathcal{E}} = (\mathcal{I} \otimes \mathcal{E}) \left( |\Phi^+\rangle\langle\Phi^+| \right),
\end{equation}
where $|\Phi^+\rangle = \frac{1}{\sqrt{d}} \sum_{i=1}^d |i\rangle \otimes |i\rangle$ is the maximally entangled state. The operator $\rho_{\mathcal{E}}$ is positive semidefinite with unit trace and satisfies the marginal constraint $\tr_2(\rho_{\mathcal{E}}) = \frac{\mathds{1}}{d}$. Conversely, every operator $\omega\succeq0$ with $\tr_2\omega=\mathds{1}/d$ is the normalized Choi matrix of a unique channel~\cite{watrous2018theory}.

To perform process tomography, we consider an IC set of pure input states
$\{{\sigma}_j\}_{j=1}^{M}$ satisfying $\frac{d}{M} \sum_{j=1}^{M} {\sigma}_j = \mathds{1}$,
together with an IC POVM $\{\frac{d}{M} P_k\}_{k=1}^{M}$ acting on the output system. If an input state ${\sigma}_j$ is prepared uniformly at random and the output is measured with the POVM, the joint outcome $(Z,Y) = (j,k)$ occurs with probability
\begin{align}
    p_{jk}
    =
    \frac{d}{M^2}
    \tr\!\left( P_k \mathcal{E}({\sigma}_j) \right)
    =
    \frac{d^2}{M^2}
    \tr\!\left( \rho_{\mathcal{E}} ({\sigma}_j^{T} \otimes P_k) \right).
\end{align}

This identity shows that process tomography is equivalent to state tomography of the Choi state $\rho_{\mathcal{E}}$ with respect to the effective POVM $\{ \frac{d^2}{M^2} \, {\sigma}_j^{T} \otimes P_k\}_{j, k=1}^{M}$. Then, all arguments developed for state tomography under non-i.i.d.\ sources apply directly, provided that at each round the implemented channel depends only on the past history. Crucially, the channel $\mathcal{E}_t$ cannot depend on the randomly chosen input state at the same round. 

\subsection{Estimation protocol}

We consider an experiment consisting of $N$ rounds. At round $t$, the device implements a channel $\mathcal{E}_t$, which may depend arbitrarily on the previous history $\mathcal{F}_{t-1}$. Our goal is to estimate the time-averaged Choi matrix $\bar{\rho}_N = \frac{1}{N} \sum_{t=1}^{N} \rho_{\mathcal{E}_t}$, corresponding to the time-averaged channel $\bar{\mathcal{E}}_N$. 

The estimation procedure consists of three steps.
\begin{enumerate}
    \item \textit{Data Collection:} For each round $t$:
    \begin{enumerate}
        \item Select an input index $Z_t$ uniformly at random and prepare ${\sigma}_{Z_t}$.
        \item Apply the channel $\mathcal{E}_t$.
        \item Measure the output using the POVM $\{ \frac{d}{M} P_k\}_{k=1}^{M}$ to obtain an outcome $Y_t$.
    \end{enumerate}

    \item \textit{Linear Estimation:} Construct a single-shot unbiased estimator $\hat{\rho}_t$ for the Choi state $\rho_{\mathcal{E}_t}$. \\
    For global $2$-design inputs and measurements:
    \begin{equation}
        \hat{\rho}_t = \left( (d+1) {\sigma}_{Z_t}^{T} - \mathds{1} \right) \otimes \left( (d+1) P_{Y_t} - \mathds{1} \right).
    \end{equation}
    For tensor products of local $2$-designs,
    \begin{align}
    \hat{\rho}_t = \left[ \bigotimes_{j=1}^n \left( 3 {\sigma}_{Z_t^j} - \mathds{1} \right)^{T} \right] \otimes \left[ \bigotimes_{j=1}^n \left( 3 P_{Y_t^j} - \mathds{1} \right) \right].
    \end{align}
    We then form the estimator $\hat{L}_N = \frac{1}{N} \sum_{t=1}^N \hat{\rho}_t$.

    \item \textit{Projection:} Project $\hat{L}_N$ onto the set of valid Choi states using the operator norm:
    \begin{align}
        \hat{\rho}_{\mathrm{PLS}} = \operatorname*{argmin}_{\omega \ge 0} \|\hat{L}_N - \omega\| \: \: \text{s.t.} \: \: \tr_2(\omega) = \frac{\mathds{1}}{d}.
        \label{eq:channel_projection}
    \end{align}
    This semidefinite program ensures the final estimate corresponds to a valid quantum channel. Since the average Choi matrix $\bar\rho_N$ is itself feasible, the triangle inequality gives $\|\hat\rho_{\mathrm{PLS}}-\bar\rho_N\|\le\|\hat\rho_{\mathrm{PLS}}-\hat L_N\|+\|\hat L_N-\bar\rho_N\|\le2\|\hat L_N-\bar\rho_N\|$ for every realization of the data.
\end{enumerate}

\subsection{Non-asymptotic guarantees}

The statistical analysis follows directly from the state tomography results, applied to the Choi state in dimension $d^2$. Using the Matrix Freedman inequality together with the appropriate variance and range bounds for the single-shot estimator (Appendix~\ref{app:channelproofs}), we obtain a non-asymptotic operator-norm bound $\| \hat{L}_N - \bar{\rho}_N \| \le \epsilon/(2d^2)$ with probability $1-\delta$. The projected operator is then guaranteed to satisfy $\| \hat{\rho}_{\mathrm{PLS}} - \bar{\rho}_N \| \le \epsilon/d^2$ with the same probability. Finally, these bounds in spectral norm are converted to bounds in diamond norm using $\Vert \Phi \Vert_{\diamond} \leq d^2 \Vert J(\Phi) \Vert$, where $J(\Phi)$ is the normalized Choi matrix of a Hermiticity-preserving map $\Phi$ \cite[Lemma~3.2]{oufkir2023sample}. 

\begin{theorem}
\label{thm:channel}
Let $\{\mathcal{E}_t\}_{t=1}^N$ be an adaptive sequence of quantum channels acting on a system of dimension $d$. For any accuracy $\epsilon > 0$ and failure probability $\delta$, the PLS estimator satisfies
\begin{align}
    \|\hat{\mathcal{E}}_{\mathrm{PLS}} - \bar{\mathcal{E}}_N\|_\diamond \le \epsilon
\end{align}
with probability $1-\delta$, provided the number of samples $N$ satisfies:\\
For global $2$-design inputs and measurements,
    \begin{align}
        N \ge \frac{32\, d^6 (1 + \epsilon/24d^2)}{\epsilon^2} \log\left(\frac{2d^2}{\delta}\right).
    \end{align}
    For local $2$-designs on $n$ qubits,
    \begin{align}
         N \ge \frac{2^{4n+3} \bigl(3^{2n} + \epsilon/6\bigr)}{\epsilon^2} \log\left(\frac{2^{2n+1}}{\delta}\right).
    \end{align}
\end{theorem}

This result confirms that PLS process tomography retains the
optimal $\tilde{\mathcal O}(d^6/\epsilon^2)$ scaling among
non-adaptive, single-copy protocols~\cite{oufkir2023sample},
even when the channel varies adaptively between rounds.

\section{Numerical illustrations}\label{sec:numerics}

\begin{figure*}[t]
\centering

\begin{minipage}[t]{0.48\textwidth}
    \centering
    \includegraphics[width=\linewidth]{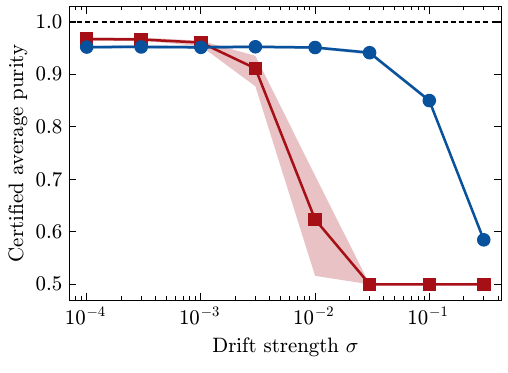}
    \caption{\label{fig:rw_purity}
    Lower bounds on the average preparation purity of the qubit source
    $(\ket0+e^{i\phi_t}\ket1)/\sqrt2$, versus the standard deviation
    $\sigma$ of its random phase increments, in radians.
    Blue circles: with tracking. Red squares: without tracking,
    using the same measurement data.
    Bounds follow Eq.~\eqref{eq:purity_certificate}, with the trivial
    floor $1/2$.
    Points show medians and shading shows interquartile ranges across
    $100$ independent runs of $N=10^{5}$ rounds each.
    The dashed line marks the purity of every preparation, which is one.}
\end{minipage}
\hfill
\begin{minipage}[t]{0.48\textwidth}
    \centering
    \includegraphics[width=\linewidth]{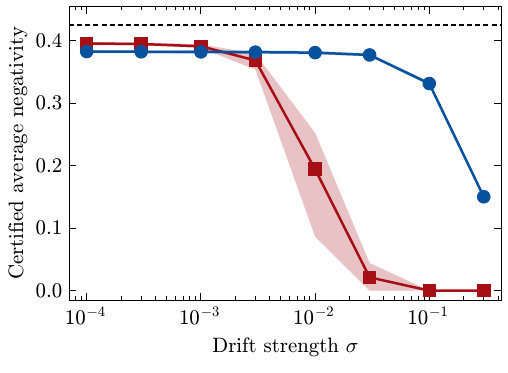}
    \caption{\label{fig:rw_bell}
    Lower bounds on the average preparation negativity of the noisy
    Bell source $v\ketbra{\Phi_{\phi_t}}+(1-v)\mathds{1}/4$, with
    $v=0.9$, versus the standard deviation $\sigma$ of its random
    phase increments, in radians.
    Blue circles: with tracking. Red squares: without tracking,
    using the same measurement data.
    Bounds are obtained from the SDP~\eqref{eq:negativity_sdp}.
    Points show medians and shading shows interquartile ranges across
    $100$ independent runs of $N=10^{5}$ rounds each.
    The dashed line marks the negativity $0.425$ of every preparation.}
\end{minipage}

\end{figure*}

We illustrate how frame tracking strengthens the purity and entanglement certificates of Sec.~\ref{sec:tracking} for sources undergoing random phase drift. For each simulated experiment, we compare processing with and without tracking using the same measurement data. Each certificate holds with probability at least $0.95$. Simulation details are given in Appendix~\ref{app:numerics}.

We consider sources whose phase evolves as
\begin{align}
    \phi_{t+1}=\phi_t+\sigma\,\xi_t,
\end{align}
where the $\xi_t$ are independent standard Gaussian variables and $\sigma$ is the standard deviation of the phase increment per round. Each experiment consists of $N=10^{5}$ rounds, with each qubit measured in a uniformly random Pauli basis. We simulate $100$ independent experiments for each value of $\sigma$.

Without tracking, we average the single-shot estimators to obtain $\hat L_N$, which estimates the average state $\bar\rho_N$. With tracking, we estimate the phase from earlier outcomes using an exponential moving average with a memory of about $50$ rounds. We then rotate each single-shot estimator by the compensating phase before averaging. The resulting estimator $\tilde L_N$ targets the average state $\tilde\rho_N$ in the moving frame. The processing has no access to the simulated phases. The certificates use only the measurement data and remain valid without assuming the simulated drift model, which affects only how informative they are.

We first consider the qubit source
\begin{align}
    \ket{\psi_t}
    =\frac{\ket0+e^{i\phi_t}\ket1}{\sqrt2},
\end{align}
so every preparation has purity one. We evaluate Eq.~\eqref{eq:purity_certificate} using $\hat L_N$ and $\tilde L_N$, including the trivial lower bound $1/2$.

Figure~\ref{fig:rw_purity} shows the results. For slow drift, $\sigma\sqrt N\ll1$, the average state is nearly pure and processing without tracking gives a median certificate of about $0.97$. The remaining gap to one is approximately the statistical correction $2\tau\|\hat L_N\|_1\approx0.03$. As the phase explores a wider range during the experiment, the average state becomes more mixed and the certificate decreases toward $1/2$.

Tracking maintains a stronger certificate over a wider range of drift rates, giving a median bound of about $0.85$ at $\sigma=0.1$. For very slow drift, it performs slightly worse than processing without tracking, giving about $0.95$ instead of $0.97$, because estimating the frame introduces statistical noise. At the largest drift, $\sigma=0.3$, tracking also becomes less informative because the phase changes appreciably over the averaging window. A longer window reduces statistical noise in the phase estimate but increases its lag behind the source.

We next consider the two-qubit source
\begin{align}
    \rho_t
    &=v\ketbra{\Phi_{\phi_t}}+(1-v)\frac{\mathds{1}}{4},
\end{align}
with $\ket{\Phi_\phi}=\frac{1}{\sqrt2}(\ket{00}+e^{i\phi}\ket{11})$ and $v=0.9$. Every preparation has negativity
$(3v-1)/4=0.425$. We apply the frame rotation to the second qubit only, so it preserves entanglement. The certificate is a lower bound obtained from the confidence region using the SDP~\eqref{eq:negativity_sdp}.

Figure~\ref{fig:rw_bell} shows that processing without tracking gives a median certificate of about $0.40$ for slow drift, but the bound falls to about $0.02$ at $\sigma=0.03$ and vanishes for stronger drift.  At $\sigma=0.3$, the true average $\bar\rho_N$ is separable in every simulated run, although every preparation is entangled.

Tracking gives median certificates of about $0.38$ up to $\sigma=0.03$, $0.33$ at $\sigma=0.1$, and $0.15$ at $\sigma=0.3$. Thus, processing the same measurements in a moving frame can certify entanglement even when $\bar\rho_N$ is separable.

\section{Conclusions}\label{sec:conclusions}

In this work, we have shown that PLS quantum state and process tomography remain statistically optimal even when the underlying quantum source exhibits adaptive time-dependent behavior. By removing the ubiquitous but physically restrictive i.i.d.\ assumption, we proved that the sample complexity required to reconstruct the time-averaged state $\bar{\rho}_N$ or channel $\bar{\mathcal{E}}_N$ matches that of the i.i.d.\ setting. Thus, temporal correlations and adaptive drift do not fundamentally increase the number of samples needed for accurate tomography.

The reconstructed state is the average of the trajectory generated during the realized experiment, which determines all affine run averages. Non-affine run averages cannot, in general, be reliably estimated from one copy of each state in an unrestricted sequence. Convex and concave functions still provide useful one-sided bounds. When earlier outcomes allow us to follow the drift, frame tracking can strengthen these bounds. Our numerical examples illustrate how it can certify average preparation purity and entanglement hidden by averaging in a fixed frame. Stronger operational goals require additional assumptions about the source or different sampling protocols. Additional structure can support predictions about future emissions, while spot checking provides a route to certifying unmeasured systems: a randomly selected subset of emissions is measured to infer properties of the remaining systems~\cite{neven2021symmetry, zhang2026efficient, zambrano2026classical, navarro2026certifying}.

An alternative approach to relaxing the i.i.d.\ assumption relies on quantum de Finetti theorems~\cite{fawzi2024learning}. While this method elegantly extends generic i.i.d.\ learning algorithms to the non-i.i.d.\ regime, the price of such universality is an increased sample complexity, scaling as $\tilde{\mathcal{O}}(d^6/\epsilon^6)$ for quantum state tomography. Furthermore, there is a fundamental conceptual difference between the methods, as the de Finetti framework targets a different effective state than the trajectory average considered here. In contrast, our analysis directly reconstructs $\bar{\rho}_N$, while retaining the optimal $\tilde{\mathcal{O}}(d r^2/\epsilon^2)$ sample complexity of PLS tomography.

Our results therefore provide a rigorous justification for applying standard PLS estimators to data collected from noisy, drifting, or actively controlled quantum hardware. In practical characterization protocols, measurement outcomes are typically aggregated into empirical frequencies under an implicit stability assumption. While such stability may fail in the presence of temporal correlations~\cite{van2013quantum}, our analysis shows that such instability does not prevent accurate reconstruction of the time-averaged object $\bar{\rho}_N$ or $\bar{\mathcal{E}}_N$.

\begin{acknowledgments}
The author is grateful to Mariana Navarro, Antonio Acín and Luciano Pereira for fruitful discussions. This work was supported by the Government of Spain (Severo Ochoa CEX2019-000910-S and FUNQIP), Fundació Cellex, Fundació Mir-Puig, Generalitat de Catalunya (CERCA program), the EU Quantera project Veriqtas, and the EU and Spanish AEI project QEC4QEA. ChatGPT Astra 6 was used to assist with proof development, manuscript drafting, and numerical code implementation.
The author verified the mathematical arguments, tested the code,
and reviewed the final text, and takes full responsibility
for the work.
\end{acknowledgments}

\bibliography{bib}

\appendix

\section{Statistical tools}\label{app:tools}

In the context of quantum state tomography via least-squares estimation, we model the sequence of estimators as a matrix-valued stochastic process. The appropriate mathematical framework for analyzing such sequences, where the state at time $t$ may depend on the previous history, is that of matrix martingales \cite{williams1991probability, tropp2011freedman}.

Formally, let $(\Omega, \mathcal{F}, P)$ be a probability space, and let
$\mathcal{F}_0 \subset \mathcal{F}_1 \subset \dots \subset \mathcal{F}_N \subset \mathcal{F}$
be a filtration. A sequence of finite-dimensional matrices $\{M_t\}_{t=0}^N$ is a matrix martingale with respect to the filtration ${\mathcal{F}_t}$ if it is adapted to the filtration (that is, $M_t$ is measurable with respect to $\mathcal{F}_t$, meaning it is determined at round $t$) and satisfies:
\begin{enumerate}
    \item $\mathbb{E}[M_t \mid \mathcal{F}_{t-1}] = M_{t-1}$,
    \item $\mathbb{E}[\|M_t\|] < \infty$ for all $t$,
\end{enumerate}
where $\mathbb{E}[ \cdot \mid \mathcal{F}_{t-1}]$ denotes conditional expectation and $\|\cdot\|$ the spectral norm. The first condition tells us that, although the magnitude of each fluctuation may depend on the past (in our setting, due to adaptive state preparation), there is no systematic drift in the conditional mean.

For simplicity, we assume $M_0 = 0$. We define the difference sequence $X_t$ as
\begin{align}
X_t = M_t - M_{t-1}.
\end{align}

To bound tail probabilities, we rely on the Matrix Freedman inequality.

\begin{theorem}[Matrix Freedman inequality~\cite{tropp2011freedman}]
\label{thm:freedman}
Let $\{M_t\}_{t=0}^N$ be a self-adjoint matrix martingale in $\mathbb{C}^{d\times d}$ adapted to a filtration $\{\mathcal{F}_t\}$, with $M_0=0$. Define the difference sequence $X_t = M_t - M_{t-1}$ and assume almost surely
\begin{align}
\|X_t\| \le R \quad \text{for all } t.
\end{align}
Define the predictable quadratic variation
\begin{align}
W_N = \sum_{t=1}^N \mathbb{E}[X_t^2 \mid \mathcal{F}_{t-1}].
\end{align}
Assume that $\|W_N\| \le \sigma^2$ almost surely. Then, for all $\tau \ge 0$,
\begin{align}
P\!\left(\|M_N\| \ge \tau\right)
\le 2d \exp\!\left(
-\frac{\tau^2/2}{\sigma^2 + R\tau/3}
\right).
\end{align}
\end{theorem}

\section{Proofs of Proposition~\ref{prop:concentration} and Theorem \ref{thm:main_result}}\label{app:proofs}

In this appendix, we provide the detailed proof of Theorem~\ref{thm:main_result}.
We first establish the martingale structure underlying the estimation procedure, then compute the bounds $R$ and $\sigma^2$ for $2$-designs and tensor products of local $2$-designs, and finally apply the Matrix Freedman inequality together with the projection lemma to obtain the stated error bounds.

\subsection{Martingale structure and unbiasedness}

We now show that the estimation error admits a natural martingale structure.

Let $\mathcal{F}_{t-1}$ denote the $\sigma$-algebra generated by the full experimental history up to step $t-1$. This includes the sequence of measurement settings and outcomes, as well as any internal randomness or external parameters used by the source. Because the state $\rho_t$ prepared at step $t$ is determined entirely by this past history, it is a function of these variables. Consequently, $\rho_t$ is $\mathcal{F}_{t-1}$-measurable, meaning it is fixed conditional on the history up to round $t-1$.

At step $t$, we perform an informationally complete measurement $\{\Pi_k\}_{k=1}^M$ on $\rho_t$. This defines a random variable $Y_t \in \{1,\dots,M\}$ with
\begin{align}
P(Y_t = k \mid \mathcal{F}_{t-1})
= \mathrm{tr}(\Pi_k \rho_t).
\end{align}
For local measurements on $n$ qubits, the outcomes are denoted by a multi-index $k=k_1, \dots, k_n$.

From the observed outcome $Y_t$, we construct a single-shot estimator $\hat{\rho}_t$. For the measurement schemes considered in this work, the single-shot estimators are given by
\begin{align}
\hat{\rho}_t = (d+1)P_{Y_t} - \mathds{1},
\end{align}
for global $2$-design measurements, and
\begin{align}
\hat{\rho}_t
= \bigotimes_{j=1}^n \left(3 P_{Y_t^j} - \mathds{1}\right),
\end{align}
for tensor products of single-qubit $2$-designs, where $Y_t = Y_t^1\dots Y_t^n$ (see Appendix \ref{app:estimators}).

The single-shot estimator is unbiased: for any fixed state $\rho$,
\begin{align}
\mathop{\mathbb{E}}_{k \sim p(\rho)}[\hat{\rho}] = \rho.
\end{align}
Conditioning on $\mathcal{F}_{t-1}$ fixes the state $\rho_t$ at round $t$, so
\begin{align}
\mathbb{E}[\hat{\rho}_{t} \mid \mathcal{F}_{t-1}]
= \mathop{\mathbb{E}}_{k \sim p(\rho_t)}[\hat{\rho}_{t}]
= \rho_{t}.
\end{align}

We therefore define the martingale difference sequence
\begin{align}
{X}_{t} = \hat{\rho}_t - \rho_t.
\end{align}
By construction, it is centered:
\begin{align}
\mathbb{E}[{X}_{t} \mid \mathcal{F}_{t-1}] = \mathbb{E}[\hat{\rho}_{t} \mid \mathcal{F}_{t-1}] - \rho_{t}
= 0.
\end{align}
The accumulated error $M_N = \sum_{t=1}^{N} {X}_{t}$ is therefore a matrix martingale adapted to the filtration  $\{\mathcal{F}_t\}$.

Our goal is to control the deviation of the linear estimator $\hat{L}_{N}  = \frac{1}{N}\sum_{t=1}^{N} \hat{\rho}_t$
from the trajectory average $\bar{\rho}_N = \frac{1}{N}\sum_{t=1}^{N} \rho_t$.
We observe that
\begin{align}
\hat{L}_{N}  - \bar{\rho}_N
= \frac{1}{N} \sum_{t=1}^{N} (\hat{\rho}_t - \rho_t)
= \frac{1}{N}M_N.
\end{align}
Thus, concentration bounds for $M_N$ directly yield bounds for the estimation error.

\subsection{Concentration bounds for global 2-design measurements}

Consider a POVM $\{\Pi_k\}_{k=1}^M$ defined by $\Pi_k = \frac{d}{M}P_k$, where the set of rank-$1$ projectors $\{P_k\}_{k=1}^M$ forms a complex projective $2$-design. As shown in Appendix \ref{app:estimators}, if the outcome $Y_t = k$ corresponding to projector $P_k$ is obtained at time $t$, the single-shot estimator takes the value
\begin{align}
\hat{\rho}_t = (d+1)P_k - \mathds{1}.
\end{align}

We first bound the operator norm of the difference sequence $X_t = \hat{\rho}_t - \rho_t$. The single-shot estimator has eigenvalues $d$ and $-1$. Since $0\preceq\rho_t\preceq\mathds{1}$, the spectrum of $X_t$ lies in $[-2,d]$, and therefore $\|X_t\|\le d$ for $d\ge2$. Thus, we may take $R = d$.

For the variance, since conditioning on the history $\mathcal{F}_{t-1}$ fixes $\rho_t$, we have
\begin{align}
    \mathbb{E}[X_t^2 \mid \mathcal{F}_{t-1}]
    &= \mathop{\mathbb{E}}_{k \sim p(\rho_t)}[(\hat{\rho}_t - \rho_t)^2] \nonumber \\
    &= \mathop{\mathbb{E}}_{k \sim p(\rho_t)}\,[\hat{\rho}_t^2] - \rho_t^2.
\end{align}
For $2$-designs, the second moment is \cite{guta2020fast}
\begin{align}
    \mathop{\mathbb{E}}_{k \sim p(\rho_t)}\,[\hat{\rho}_t^2] = (d-1)\rho_t + d\,\mathds{1}.
\end{align}
Substituting this back and taking the spectral norm, we obtain a bound valid for all $t$:
\begin{align}
  \bigl\| \mathbb{E}[X_t^2 \mid \mathcal{F}_{t-1}] \bigr\|
  = \|(d-1)\rho_t + d\,\mathds{1} - \rho_t^2\|
  \le 2d.
\end{align}
Hence, $\|W_N\| = \Vert \sum_{t=1}^N \mathbb{E}[X_t^2 \mid \mathcal{F}_{t-1}] \Vert\le 2dN$.

Applying the Matrix Freedman inequality with $\tau = N\epsilon$, we obtain
\begin{align}
    P\!\left(\|\hat{L}_N - \bar{\rho}_N\| \ge \epsilon\right)
    \le 2d \exp\!\left( -\frac{N\epsilon^2}{4d\left(1 + \epsilon/6\right)} \right),
\end{align}
which is Eq.~\eqref{eq:concentration} with $(v,R)=(2d,d)$.

\subsection{Concentration bounds for local 2-design measurements}

Let the measurement consist of a POVM constructed from local $2$-designs. This means that for each qubit $j=1,\dots,n$, we have a local $2$-design defined by rank-1 projectors $\{P_{k}\}_{k=1}^m$. Then, the POVM elements are $\Pi_{k} = \frac{d}{M} P_{k}$, where $P_{k} = \bigotimes_{j=1}^n P_{k_j}$, $M=m^n$, and the outcome $Y_t$ at time $t$ takes values $Y_t = k_1 \dots k_n$.

As derived in Appendix~\ref{app:estimators}, if outcome $Y_t$ is observed at time $t$, the corresponding single-shot estimator takes the form
\begin{align}
 \hat{\rho}_{t}
    = \bigotimes_{j=1}^n \left( 3 P_{Y_t^{j}} - \mathds{1} \right).
\end{align}

The martingale difference sequence is $X_t = \hat{\rho}_t - \rho_t$. Each local factor $3P-\mathds{1}$ has eigenvalues $2$ and $-1$, so the eigenvalues of $\hat\rho_t$ are products of these numbers and lie in $[-2^{n-1},2^n]$. Since $0\preceq\rho_t\preceq\mathds{1}$, the spectrum of $X_t$ lies in $[-2^{n-1}-1,2^n]$, and
\begin{align}
     \Vert X_t \Vert \leq 2^n = d.
\end{align}
We then set $R = d$.

We now compute the conditional second moment. Since $\rho_t$ is fixed given $\mathcal{F}_{t-1}$,
we have $\mathbb{E}[{X}_{t}^2 \mid \mathcal{F}_{t-1}]
    = \mathop{\mathbb{E}}_{k \sim p(\rho_t)} [\hat{\rho}_t^2]
    - \rho_t^2 \preceq \mathop{\mathbb{E}}_{k \sim p(\rho_t)} [\hat{\rho}_t^2]$.

For a projector $P$, $(3P-\mathds{1})^2 = 3P + \mathds{1}$. Then,
\begin{align}
    \mathop{\mathbb{E}}_{k \sim p(\rho_t)} [\hat{\rho}_t^2] &= \sum_{k} \frac{d}{M} \tr(P_{k} \rho_{t})  \bigotimes_{j=1}^n (3 P_{k_j} - \mathds{1})^2 \nonumber \\
    &= \frac{d}{M} \sum_{k} \tr(P_{k} \rho_{t}) \bigotimes_{j=1}^n (3 P_{k_j} + \mathds{1}).
\end{align}
For a product state $\rho_t = \rho_{1} \otimes \dots \otimes \rho_{n}$, the sum factorizes into a product of local sums. For each qubit, we have
\begin{align}
    \frac{2}{m} \sum_{k_j=1}^m & \tr(P_{k_j}\rho_j) (3P_{k_j} + \mathds{1}) \nonumber\\
    =  \frac{2}{m} & \left[ 3 \sum_{k_j=1}^{m} \tr(P_{k_j} \rho_j)P_{k_j} + \sum_{k_j=1}^{m} \tr(P_{k_j} \rho_j)\mathds{1} \right].
\end{align}
Using the 2-design properties $\sum \tr(P_{k_j} \rho_j) P_{k_j} = \frac{m}{6}(\rho_j+\mathds{1})$ and $\sum \tr(P_{k_j} \rho_j) = \frac{m}{2}$, we obtain
\begin{align}
    \frac{2}{m} \sum_{k_j=1}^m \tr(P_{k_j} \rho_j) (3P_{k_j} + \mathds{1}) &= \rho_j + 2\mathds{1}.
\end{align}

The global second moment is the tensor product of these local terms:
\begin{align}
    \mathop{\mathbb{E}}_{k \sim p(\rho_t)} [\hat{\rho}_t^2]  &= \bigotimes_{j=1}^n (\rho_j + 2\mathds{1}) \nonumber \\
    & = \sum_{\alpha \in \mathcal{P}([n])} 2^{|\alpha|} \mathrm{tr}_\alpha (\rho_t) \otimes \mathds{1}^{\otimes \alpha},
\end{align}
where $\mathrm{tr}_\alpha(\rho_t)$ denotes the partial trace over the subsystems indexed by $\alpha$ and $\mathcal{P}([n])$ is the power set of $\{1, \dots, n\}$. By linearity, the same identity holds for arbitrary (not necessarily product) states.

Using $\|\operatorname{tr}_{\alpha}(\rho_t)\| \le 1$, we obtain
\begin{align}
    \left\|\mathop{\mathbb{E}}_{k \sim p(\rho_t)} \, [\hat{\rho}_t^2 ]\right\|
    \le \sum_{\alpha \in \mathcal{P}([n])} 2^{|\alpha|}
    = 3^n.
\end{align}
Therefore, for $N$ steps, $\|W_N\| \leq \sum_{t=1}^N \Vert \mathbb{E}[X_t^2\mid\mathcal{F}_{t-1}]\Vert \le N 3^n$.

We finally apply the Matrix Freedman inequality:
\begin{align}
    P\!\left(\|\hat{L}_N - \bar{\rho}_N\| \ge \epsilon\right)
    \le 2d \exp\!\left(
    -\frac{N\epsilon^2}
    {2\left(3^n + 2^n \epsilon /3\right)}
    \right),
\end{align}
which is Eq.~\eqref{eq:concentration} with $(v,R)=(3^n,2^n)$. This proves Proposition~\ref{prop:concentration}.

\subsection{Proof of Theorem~\ref{thm:main_result}}

By Lemma~\ref{lem:norm_conversion} with $\sigma=\bar\rho_N$, it suffices to ensure that the spectral deviation of the linear estimator satisfies $\|\hat L_N-\bar\rho_N\|\le\tau$ with
\begin{equation}\label{eq:tau_condition}
\tau = \frac{\epsilon}{4r}.
\end{equation}

From the bound derived in Appendix~\ref{app:proofs} for $2$-designs,
\begin{align}
P\!\left(\|\hat{L}_N - \bar{\rho}_N\| \ge \tau\right)
\le
2d \exp\!\left(
-\frac{N\tau^2}{4d(1 + \frac{\tau}{6})}
\right).
\end{align}
Substituting $\tau = \epsilon/(4r)$ yields
\begin{align}
P\!\left(\|\hat{L}_N - \bar{\rho}_N\| \ge \tau\right)
\le
2d \exp\!\left(
-\frac{N\epsilon^2}{64 d r^2 (1 + \frac{\epsilon}{24r})}
\right).
\end{align}

Therefore, if
\begin{equation}
N \ge
\frac{64 \, d \, r^2 (1 + \epsilon/24r)}{\epsilon^2}
\log\!\left(\frac{2d}{\delta}\right),
\end{equation}
then with probability at least $1-\delta$,
\begin{align}
\|\hat{\rho}_{\mathrm{PLS}} - \bar{\rho}_N\|_1
\le
\epsilon
+ 2 \min\!\left\{
\Sigma_r(\bar{\rho}_N),
\Sigma_r(\hat{\rho}_{\mathrm{PLS}})
\right\}.
\end{align}

For tensor products of local $2$-designs, the Matrix Freedman inequality gives
\begin{align}
P\!\left(\|\hat{L}_N - \bar{\rho}_N\| \ge \tau\right)
\le
2d
\exp\!\left(
-\frac{N\tau^2}{2(3^n + \frac{2^{n}\tau}{3})}
\right).
\end{align}
Substituting again $\tau = \epsilon/(4r)$ yields
\begin{align}
P\!\left(\|\hat{L}_N - \bar{\rho}_N\| \ge \tau\right)
\le
2d
\exp\!\left(
-\frac{N\epsilon^2}{32 r^2 (3^n + \frac{2^{n}\epsilon}{12r})}
\right).
\end{align}

Consequently, if
\begin{align}
N
\ge
\frac{32 \, r^2 \bigl(3^n + 2^n \epsilon/12r\bigr)}{\epsilon^2}
\log\!\left(\frac{2^{n+1}}{\delta}\right),
\end{align}
then with probability at least $1-\delta$,
\begin{align}
\|\hat{\rho}_{\mathrm{PLS}} - \bar{\rho}_N\|_1
\le
\epsilon
+ 2 \min\!\left\{
\Sigma_r(\bar{\rho}_N),
\Sigma_r(\hat{\rho}_{\mathrm{PLS}})
\right\}.
\end{align}

\qed

\section{Least-squares estimators for state tomography}\label{app:estimators}

Here, following Ref.~\cite{guta2020fast}, we re-derive the explicit form of the least-squares estimators for quantum states when the measurement ensemble forms either a global or a local complex projective 2-design.

\subsection{Estimator for global 2-designs}

\begin{theorem}
Let $\rho$ be a quantum state on a $d$-dimensional Hilbert space. Consider a measurement described by a set of POVM elements $\{\Pi_k\}_{k=1}^M$, where $\Pi_k = \frac{d}{M} P_k$ and $\{P_k\}_{k=1}^M$ is a set of rank-1 projectors forming a complex projective 2-design. The least-squares estimator of $\rho$, based on empirical relative frequencies $\hat{f}_k$, is
\begin{align}
\hat{L}_N = \sum_{k=1}^M \hat{f}_k \left[(d+1) P_k - \mathds{1} \right].
\end{align}
\end{theorem}

\begin{proof}
The probability of obtaining outcome $k$ given the state $\rho$ is $p_k = \tr(\Pi_k \rho) = \frac{d}{M} \tr(P_k \rho)$. We define the linear map $\mathcal{A}: \rho \mapsto \{p_k\}_{k=1}^M$. The least-squares estimator minimizes $\|\mathcal{A}(\rho) - \hat{f}\|^2$, yielding the solution
\begin{align}
\hat{L}_N = (\mathcal{A}^\dagger \mathcal{A})^{-1} \mathcal{A}^\dagger (\hat{f}),
\end{align}
where the adjoint map is $\mathcal{A}^\dagger(q) = \frac{d}{M} \sum_k q_k P_k$.

Next, we compute the composition $\mathcal{A}^\dagger \mathcal{A}$. For an arbitrary operator $X$,
\begin{align}
\mathcal{A}^\dagger \mathcal{A} (X)
= \frac{d^2}{M^2} \sum_{k=1}^M \tr(P_k X) P_k.
\end{align}
Since $\{P_k\}_{k=1}^{M}$ forms a 2-design, it exactly reproduces the first and second moments of the Haar measure. Then,
\begin{align}
\frac{1}{M}\sum_{k=1}^{M} \tr(P_k X) P_k = \frac{1}{d(d+1)} \left(X + \tr(X)\mathds{1}\right).
\end{align}
Substituting,
\begin{align}
\mathcal{A}^\dagger \mathcal{A} (X) =  \frac{d}{M(d+1)} \left(X + \tr(X)\mathds{1}\right) =\frac{d}{M} \mathcal{D}_{p}(X),
\end{align}
where $\mathcal{D}_p(X) = p X + (1-p) \tr(X) \frac{\mathds{1}}{d}$ is the depolarizing channel with parameter $p = \frac{1}{d+1}$.

The inverse of this map is proportional to the inverse of the depolarizing channel. Specifically, $(\mathcal{A}^\dagger \mathcal{A})^{-1}(Y) = \frac{M}{d} \mathcal{D}_p^{-1}(Y)$. The inverse depolarizing channel is given by $\mathcal{D}_p^{-1}(Y) = (d+1)Y - \tr(Y)\mathds{1}$.

Applying this inverse to $\mathcal{A}^\dagger(\hat{f}) = \frac{d}{M} \sum_k \hat{f}_k P_k$ we obtain
\begin{align}
\hat{L}_N
&= (\mathcal{A}^\dagger \mathcal{A})^{-1} \left( \frac{d}{M} \sum_{k=1}^M \hat{f}_k P_k \right) \nonumber \\
&= \sum_{k=1}^M \hat{f}_k \left( (d+1) P_k - \tr(P_k)\mathds{1} \right) \nonumber \\
&= \sum_{k=1}^M \hat{f}_k \left( (d+1) P_k - \mathds{1} \right),
\end{align}
since $\tr(P_k)=1$.
\end{proof}

\subsection{Estimator for local 2-designs}

\begin{theorem}
Let $\rho$ be a state on an $n$-qubit system ($d=2^n$). Suppose the measurement POVM elements are tensor products of single-qubit 2-designs. Specifically, let the outcomes be indexed by $\mathbf{k} = (k_1, \dots, k_n)$, corresponding to POVM elements
\begin{align}
\Pi_{\mathbf{k}} = \frac{d}{M} \bigotimes_{j=1}^n P_{k_j},
\end{align}
where each local set $\{P_{k}\}_{k=1}^{m}$ is a rank-1 projective 2-design on qubit $j$ with $m$ outcomes. The least-squares estimator of $\rho$, based on empirical relative frequencies $\hat{f}_{\mathbf{k}}$, is
\begin{align}
\hat{L}_N = \sum_{\mathbf{k}} \hat{f}_{\mathbf{k}}
\bigotimes_{j=1}^n \left( 3 P_{k_j} - \mathds{1} \right).
\end{align}
\end{theorem}

\begin{proof}

Let $d=2^n$ be the Hilbert space dimension and $M = m^n$ the total number of measurement outcomes, where $m$ is the number of elements in each local 2-design. Define $P_{\mathbf{k}} = \bigotimes_{j=1}^n P_{k_j}$,
so that $\Pi_{\mathbf{k}} = \frac{d}{M} P_{\mathbf{k}}$.

The probability of obtaining outcome $\mathbf{k}$ is $p_{\mathbf{k}} = \tr(\Pi_{\mathbf{k}} \rho)
= \frac{d}{M} \tr(P_{\mathbf{k}} \rho)$. We define the measurement map $\mathcal{A}: \rho \mapsto \{p_{\mathbf{k}}\}$. As before, the least-squares estimator is
\begin{align}
\hat{L}_N = (\mathcal{A}^\dagger \mathcal{A})^{-1} \mathcal{A}^\dagger(\hat{f}),
\end{align}
with $\mathcal{A}^\dagger(\hat{f}) = \frac{d}{M} \sum_{\mathbf{k}} \hat{f}_{\mathbf{k}} P_{\mathbf{k}}$.

We now compute $\mathcal{A}^\dagger \mathcal{A}$. For a product operator
$X = X_1 \otimes \dots \otimes X_n$, we have
\begin{align}
\mathcal{A}^\dagger \mathcal{A} (X)
&= \sum_{\mathbf{k}} \tr(\Pi_{\mathbf{k}} X) \Pi_{\mathbf{k}} \nonumber \\
&= \frac{d^2}{M^2}
\bigotimes_{j=1}^n
\left(
\sum_{k_j=1}^m
\tr(P_{k_j} X_j) P_{k_j}
\right).
\end{align}

Using the single-qubit 2-design identity (dimension $2$),
\begin{align}
\sum_{k_j=1}^m
\tr(P_{k_j} X_j) P_{k_j}
= \frac{m}{6}
\left( X_j + \tr(X_j)\mathds{1} \right)
= \frac{m}{2}\, \mathcal{D}_{1/3}(X_j),
\end{align}
where $\mathcal{D}_{1/3}$ denotes the qubit depolarizing channel with parameter $1/3$.

Since $\frac{d}{M} = \left(\frac{2}{m}\right)^n$, the prefactors simplify as
\begin{align}
\mathcal{A}^\dagger \mathcal{A}
= \left( \frac{2}{m} \right)^{2n}
\bigotimes_{j=1}^n
\left( \frac{m}{2} \mathcal{D}_{1/3}^{(j)} \right) = \frac{d}{M}
\bigotimes_{j=1}^n \mathcal{D}_{1/3}^{(j)}.
\end{align}
By linearity, this identity extends to arbitrary operators $X$.

Therefore, $(\mathcal{A}^\dagger \mathcal{A})^{-1}
= \frac{M}{d}
\bigotimes_{j=1}^n
\left( \mathcal{D}_{1/3}^{(j)} \right)^{-1}$.
The inverse of the qubit depolarizing channel is
\begin{align}
\left( \mathcal{D}_{1/3}^{(j)} \right)^{-1}(Y)
= 3Y - \tr(Y)\mathds{1}.
\end{align}

Applying this inverse to $\mathcal{A}^\dagger(\hat{f})$, we obtain
\begin{align}
\hat{L}_N
&= (\mathcal{A}^\dagger \mathcal{A})^{-1}
\left( \frac{d}{M} \sum_{\mathbf{k}} \hat{f}_{\mathbf{k}} P_{\mathbf{k}} \right) \nonumber \\
&= \sum_{\mathbf{k}} \hat{f}_{\mathbf{k}}
\bigotimes_{j=1}^n
\left( 3 P_{k_j} - \mathds{1} \right),
\end{align}
which completes the proof.
\end{proof}

\section{Proof of Theorem~\ref{thm:first_moment_converse}}
\label{app:nonlinear_proof}

\begin{proof}
If $g$ is affine, then $g(\rho)=\tr(H\rho)+c$ for some Hermitian operator $H$ and constant $c$. Therefore,
\begin{align}
    G_N
    =
    \frac{1}{N}\sum_{t=1}^N g(\rho_t)
    =
    g(\bar\rho_N),
\end{align}
which can be estimated by reconstructing $\bar\rho_N$ with the tomography protocol of Theorem~\ref{thm:main_result}.

Suppose now that $g$ is non-affine. There exist states $\rho_0,\rho_1$ and $\lambda\in(0,1)$ such that, with $\bar\rho=\lambda\rho_0+(1-\lambda)\rho_1$,
\begin{align}
    a=g(\bar\rho)\neq \lambda g(\rho_0)+(1-\lambda)g(\rho_1)=b.
\end{align}
Let $\Delta=|a-b|$, $\eta_g=\Delta/4$, and
\begin{align}
    R_N=\sup_{\rho_1,\ldots,\rho_N}
    \Pr(|\widehat G_N-G_N|>\eta_g).
\end{align}

For the constant trajectory $\bar\rho,\ldots,\bar\rho$, the protocol receives $\bar\rho^{\otimes N}$ and the target is $a$. If $Q_N$ is the resulting distribution of $\widehat G_N$, then
\begin{align}
    Q_N(|\widehat G_N-a|\le\eta_g)\ge1-R_N.
    \label{eq:constant_success}
\end{align}

Consider instead independent hidden labels $Z_t\in\{0,1\}$ with $\Pr(Z_t=0)=\lambda$, and prepare $\rho_{Z_1}\otimes\cdots\otimes\rho_{Z_N}$. Its target is
\begin{align}
    G_N^{(Z)}
    =
    \frac{1}{N}\sum_{t=1}^N g(\rho_{Z_t}).
\end{align}
For every fixed realization $Z=z$, this is a deterministic trajectory. Hence, by the definition of $R_N$,
\begin{align}
    \Pr\left(
        |\widehat G_N-G_N^{(Z)}| \leq \eta_g
        \,\middle|\, Z=z
    \right)
    \geq 1 - R_N.
\end{align}
Averaging over the hidden labels gives
\begin{align}
    \Pr\left( |\widehat G_N-G_N^{(Z)}| \leq \eta_g \right) \geq 1 - R_N.
    \label{eq:mixture_estimation_error}
\end{align}

Let $L=|g(\rho_0)-g(\rho_1)|$. If $L>0$, Hoeffding's inequality~\cite{vershynin_2018} gives
\begin{align}
    q_N=\Pr(|G_N^{(Z)}-b|>\eta_g/2)
    \le2\exp\!\left(-\frac{N\Delta^2}{32L^2}\right),
    \label{eq:hidden_hoeffding}
\end{align}
whereas $q_N=0$ if $L=0$. Combining this with Eq.~\eqref{eq:mixture_estimation_error} yields
\begin{align}
    \Pr\left(
        |\widehat G_N-b|\le\frac{3\eta_g}{2}
    \right)
    \ge1-R_N-q_N.
    \label{eq:mixture_success}
\end{align}

After averaging over the hidden labels, however, the quantum state supplied to the protocol is
\begin{align}
    \mathbb{E}_Z\left[
        \rho_{Z_1}\otimes\cdots\otimes\rho_{Z_N}
    \right]
    =
    \bar\rho^{\otimes N}.
\end{align}
Therefore the distribution in Eq.~\eqref{eq:mixture_success} is also $Q_N$.

The intervals centered at $a$ and $b$ have respective radii $\eta_g$ and $3\eta_g/2$ and are disjoint because $5\eta_g/2=5\Delta/8<\Delta$. Combining Eqs.~\eqref{eq:constant_success} and~\eqref{eq:mixture_success} therefore yields
\begin{align}
    1 \ge 2(1-R_N)-q_N.
\end{align}
Consequently, if $L>0$,
\begin{align}
    R_N
    \ge
    \frac{1}{2}
    -
    \exp\left(
        -\frac{N\Delta^2}{32L^2}
    \right),
\end{align}
while for $L=0$ one has $R_N\ge1/2$. This proves the theorem.
\end{proof}

\section{Proof of Proposition~\ref{prop:tracking}}
\label{app:newproofs}

\begin{proof}[Proof of Proposition~\ref{prop:tracking}]
Because $V_t$ is chosen from the observed history before round
$t$, it is $\mathcal{F}_{t-1}$-measurable. Conditional
unbiasedness therefore gives
\begin{align}
    \mathbb{E}\!\left[
        V_t\hat\rho_tV_t^\dagger
        \,\middle|\,\mathcal{F}_{t-1}
    \right]
    =
    V_t\rho_tV_t^\dagger.
\end{align}
Then $\tilde X_t=V_tX_tV_t^\dagger$ is a martingale difference
sequence, and
\begin{align}
    \tilde L_N-\tilde\rho_N
    =\frac{1}{N}\sum_{t=1}^N\tilde X_t.
\end{align}

Unitary conjugation preserves the operator norm, so
$\|\tilde X_t\|=\|X_t\|\le R$. It also gives
\begin{align}
    \mathbb{E}[\tilde X_t^2\mid\mathcal{F}_{t-1}]
    =
    V_t\,
    \mathbb{E}[X_t^2\mid\mathcal{F}_{t-1}]
    \,V_t^\dagger,
\end{align}
whose norm is at most $v$.
The concentration argument of Appendix~\ref{app:proofs}
therefore applies with the same constants.
Combining it with Lemma~\ref{lem:norm_conversion} proves the
tracked version of Theorem~\ref{thm:main_result}.

Finally, suppose $g$ is convex and invariant under the chosen
rotations, so that $g(V_t\rho V_t^\dagger)=g(\rho)$.
Jensen's inequality gives
\begin{align}
    \frac{1}{N}\sum_{t=1}^N g(\rho_t)
    =
    \frac{1}{N}\sum_{t=1}^N
    g(V_t\rho_tV_t^\dagger)
    \ge
    g(\tilde\rho_N).
\end{align}
This proves Eq.~\eqref{eq:tracked_jensen}.
\end{proof}

\section{Range and variance for process tomography}
\label{app:channelproofs}

For global $2$-design inputs and measurements, the single-shot
estimator is
\begin{align}
    \hat\rho_t
    =
    \left((d+1)\sigma_{Z_t}^T-\mathds{1}\right)
    \otimes
    \left((d+1)P_{Y_t}-\mathds{1}\right).
\end{align}
Each factor has eigenvalues $d$ and $-1$, so $\hat\rho_t$
has eigenvalues $d^2$, $-d$, and $1$.
Since $0\preceq\rho_{\mathcal E_t}\preceq\mathds{1}$,
the spectrum of $X_t=\hat\rho_t-\rho_{\mathcal E_t}$ lies
in $[-d-1,d^2]$. Thus $\|X_t\|\le d^2$, and we take $R=d^2$.

To bound the variance, we use
$((d+1)P-\mathds{1})^2=(d^2-1)P+\mathds{1}$.
Applying the $2$-design identity to each factor:
\begin{align}
    \mathbb E[\hat\rho_t^2\mid\mathcal F_{t-1}]
    ={}&
    (d-1)^2\rho_{\mathcal E_t}
    +d^2\,\mathds{1}\otimes\mathds{1}
    \nonumber\\
    &+d(d-1)\left[
        \tr_2(\rho_{\mathcal E_t})\otimes\mathds{1}
        +\mathds{1}\otimes\tr_1(\rho_{\mathcal E_t})
    \right].
\end{align}
The Choi constraint
$\tr_2\rho_{\mathcal E_t}=\mathds{1}/d$ and the bounds
$\|\rho_{\mathcal E_t}\|\le1$ and
$\|\tr_1\rho_{\mathcal E_t}\|\le1$ imply
\begin{align}
    \left\|
        \mathbb E[\hat\rho_t^2\mid\mathcal F_{t-1}]
    \right\|
    \le 3d^2-2d.
\end{align}
Conditional unbiasedness also gives
\begin{align}
    \mathbb E[X_t^2\mid\mathcal F_{t-1}]=
    \mathbb E[\hat\rho_t^2\mid\mathcal F_{t-1}]
    -\rho_{\mathcal E_t}^{\,2}\preceq
    \mathbb E[\hat\rho_t^2\mid\mathcal F_{t-1}].
\end{align}
We may therefore take $v=4d^2$.

For local $2$-designs on $n$ input and $n$ output qubits,
the estimator is a tensor product of $2n$ single-qubit
factors. The bounds of Appendix~\ref{app:proofs} apply
directly and give
\begin{align}
    (v,R)=(3^{2n},2^{2n}).
\end{align}

Substituting these constants into Eq.~\eqref{eq:concentration},
in dimension $d^2$ and with $\tau=\epsilon/(2d^2)$, gives
the required concentration bound with the sample counts of
Theorem~\ref{thm:channel}.

\section{Details of the numerical simulations}
\label{app:numerics}

We simulate the qubit and Bell sources of Sec.~\ref{sec:numerics}
with $N=10^{5}$ rounds. The phase starts uniformly on $[0,2\pi)$
and evolves as $\phi_{t+1}=\phi_t+\sigma\xi_t$, where the $\xi_t$
are independent standard Gaussian variables. For each drift
strength shown in the figures, we generate $100$ independent
walks and use the same walks for both sources. The phases do
not depend on the measurements and are not available to the
data processing.

Each qubit is measured independently in a uniformly random
Pauli basis $P\in\{X,Y,Z\}$. An outcome $s=\pm1$ gives the
single-shot estimator $(\mathds{1}+3sP)/2$, with tensor products
for two qubits. Averaging these estimators gives $\hat L_N$.

We construct a noisy signal whose conditional mean
is $e^{i\phi_t}$ for the qubit source and $v e^{i\phi_t}$
for the Bell source:
\begin{align}
    q_t&=3s_t\left(\mathds{1}_X+i\mathds{1}_Y\right),
    \\
    q_t&=\frac92 a_tb_t
    \left(
        \mathds{1}_{XX}-\mathds{1}_{YY}
        +i\mathds{1}_{XY}+i\mathds{1}_{YX}
    \right).
\end{align}
Here $s_t=\pm1$ is the qubit outcome, $a_t,b_t=\pm1$
are the local Bell-source outcomes, and the indicators
select the measured settings. For example,
$\mathds{1}_{XY}=1$ when the first qubit is measured in $X$
and the second in $Y$.

We average these signals using
\begin{align}
    z_t=(1-\alpha)z_{t-1}+\alpha q_t,
    \qquad \alpha=0.02,\quad z_0=0,
\end{align}
which gives more weight to recent outcomes.
For round $t$, the estimated phase is
$\theta_t=\arg z_{t-1}$, taking $\theta_t=0$ when
$z_{t-1}=0$. We compensate this phase by conjugating
the single-shot estimator with
$V_t=\operatorname{diag}(1,e^{-i\theta_t})$,
acting only on the second qubit for the Bell source.
Averaging the rotated matrices gives $\tilde L_N$.
Since each rotation uses only earlier outcomes, the
guarantees of Sec.~\ref{sec:tracking} apply.

For both processing rules, we choose the confidence radius
from Proposition~\ref{prop:concentration}, using
$(v,R)=(3^n,2^n)$ and failure probability $\delta=0.05$.
This gives $\tau\simeq0.016$ for one qubit and
$\tau\simeq0.030$ for two qubits. Each displayed certificate
therefore holds with probability at least $95\%$.

We evaluate the purity certificate~\eqref{eq:purity_certificate}
and the negativity SDP~\eqref{eq:negativity_sdp} using either
$\hat L_N$ or $\tilde L_N$. For purity, we also impose the
trivial lower bound $1/2$. No projection onto density matrices
is needed. For negativity, we solve the dual SDP and check
dual feasibility and agreement with the primal value
numerically.

\end{document}